\documentclass[conference]{IEEEtran}

\usepackage[caption=false,font=footnotesize]{subfig}
\usepackage{cite}
\usepackage{amsmath,amssymb,amsfonts,amsthm}
\usepackage{bm}
\usepackage{graphicx}
\usepackage{textcomp}
\usepackage{microtype}
\usepackage{csquotes}
\usepackage{xcolor}
\usepackage{comment}
\usepackage{float}
\usepackage{mathtools}
\usepackage[ruled,linesnumbered]{algorithm2e}

\SetCommentSty{mycommfont}

\def\BibTeX{{\rm B\kern-.05em{\sc i\kern-.025em b}\kern-.08em
    T\kern-.1667em\lower.7ex\hbox{E}\kern-.125emX}}

\SetKw{Continue}{continue}
\SetKw{Break}{break}
\SetKwInOut{Input}{Input}
\SetKwInOut{Output}{Output}

\newtheorem{theorem}{Theorem}
\newtheorem{lemma}{Lemma}

\newtheorem{definition}{Definition}
\newtheorem{example}{Example}
\newtheorem{principle}{Principle}

\newcommand{\mypath}{\lambda}
\newcommand{\chain}{\gamma}
\newcommand{\pv}{\pi} 

\graphicspath{{./graphics/}}

\begin{document}

\title{
Longer Is Shorter: Making Long Paths to Improve the Worst-Case Response Time of DAG Tasks}

\author{Qingqiang He\textsuperscript{1},
Nan Guan\textsuperscript{2},
Mingsong Lv\textsuperscript{1}
\\
\\
\textsuperscript{1}The Hong Kong Polytechnic University, China\\
\textsuperscript{2}City University of Hong Kong, China\\
}

\maketitle
\begin{abstract}\
DAG (directed acyclic graph) tasks are widely used to model parallel real-time workload.
The real-time performance of a DAG task not only depends on its total workload, but also its graph structure. Intuitively, with the same
total workload, a DAG task with looser precedence constraints tends to have better real-time performance in terms of worst-case response time.
However, this paper shows that actually we can shorten the worst-case response time of a DAG task by carefully adding new edges and constructing longer paths.
We develop techniques based on the state-of-the-art DAG response time analysis techniques to properly add new edges so that the worst-case response time bound guaranteed by formal analysis can be significantly reduced.
Experiments under different parameter settings demonstrate the effectiveness of the proposed techniques.
\end{abstract}


\section{Introduction}
\label{sec:introduction}
More and more real-time applications are parallelized to execute on multi-core processors for high performance and energy efficiency.
DAG (directed acyclic graph) is a widely used model to describe the structure constraints of parallel real-time tasks.
As an example, a processing chain from perception to control in the autonomous driving system can be modeled as a sporadic DAG task \cite{verucchi2020latency}.
There have been a large number of research works on real-time scheduling and analysis of DAG tasks in recent years \cite{ueter2018reservation, chen2019timing, han2019response, jiang2020real, ueter2021response, zhao2022dag, osborne2022minimizing, ueter2022parallel, bi2022response, lin2022type}, where a fundamental problem to solve is how to upper-bound the worst-case response time of a DAG task executing on a parallel processing platform.

The worst-case response time of a DAG task depends on its graph structure characteristics.
Intuitively, given the same total workload, a DAG task with looser precedence constraints among the vertices tends to have a shorter response time,
as its workload has a better chance to be executed in parallel and thus utilize computing resources better.
On the contrary, a DAG task with stricter precedence constraints tends to have a larger response time, as the workload has to be executed more sequentially. It seems that enforcing more precedence constraints on a DAG task is detrimental to its responsiveness.

However, this paper shows that actually we can improve (i.e., shorten) the worst-case response time of a DAG task by carefully enforcing extra precedence constraints, i.e., adding new edges to the original DAG\footnote{Note that adding new edges does not require changing the task itself. This can be achieved by, e.g., letting the scheduler be aware and enforce the corresponding precedence constraints when scheduling the vertices.}.
The key observation is that, by properly adding new edges and thus constructing some longer paths, we can reduce the worst-case interference to its critical path, which is the bottleneck for the DAG task to finish execution, and thus shorten the worst-case response time. The challenge is how to find the right edges to add so that the worst-case response time bound guaranteed by formal analysis is indeed improved.

In this paper, we develop techniques based on the above observation to improve the worst-case response time bound guaranteed by formal analysis.
More specifically, we identify the principles of adding edges by carefully examining the dependencies inside a DAG task and propose a simple but rather effective method to add edges based on the state-of-the-art worst-case response time analysis techniques in \cite{he2022bounding}.
As pointed out in \cite{zhao2022dag}, DAG tasks are to model the two major characteristics of parallel applications: \emph{parallelism} and \emph{dependency}.
The work in \cite{he2022bounding} utilizes long paths to explore the parallelism inside DAG tasks to improve system schedulability.
This paper is to explore the dependencies among vertices inside DAG tasks, further advancing the state-of-the-art.

We also propose a scheduling approach that applies the developed techniques to task systems consisting of multiple DAG tasks.
Experiments show that the proposed method significantly outperforms the state-of-the-art, reducing the worst-case response time bound by 21.6\% and improving the system schedulability by 22.2\% on average.

\section{System Model}
\label{sec:model}
\subsection{Task Model}
\label{sec:task}

\begin{figure}[t]
\centering
\subfloat[a DAG task]{
    \includegraphics[width=0.4\linewidth]{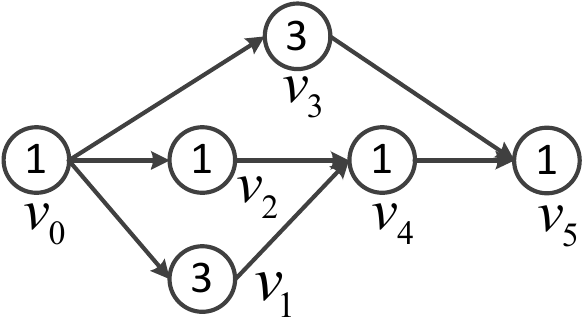}
    \label{fig:dag_example}
}
\hfil
\subfloat[a possible schedule]{
    \includegraphics[width=0.45\linewidth]{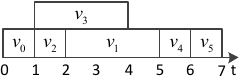}
    \label{fig:sequence1}
}
\caption{Illustration of the system model.}
\label{fig:example}
\end{figure}

We consider a parallel real-time task modeled as a DAG $G = (V, E)$, where $V$ is the set of vertices and $E\subseteq  V \times V $ is the set of edges.
Each vertex $v\in V$ represents a piece of sequentially executed workload and has a WCET (worst-case execution time) $c(v)$.
An edge $(v_i, v_j)\in E$ represents the precedence constraint between $v_i$ and $v_j$, which means that $v_j$ can only start its execution after $v_i$ completes its execution.
A vertex with no incoming edges is called a \emph{source vertex} and a vertex with no outgoing edges is called a \emph{sink vertex}.
Without loss of generality, we assume that $G$ has exactly one source  (denoted as $v_{src}$) and one sink (denoted as $v_{snk}$).
If $G$ has multiple source or sink vertices, we add a dummy source or sink vertex with zero WCET to comply with the assumption.

A \emph{path} $\mypath$ is a set of vertices $(\pv_0, \cdots, \pv_k)$ such that $\forall i\in [0,k-1]$: $(\pv_i, \pv_{i+1})\in E$.
The length of a path is the total workload in this path and is defined as $len(\mypath) \coloneqq \sum_{\pv_i\in \mypath} c(\pv_i)$.
A \emph{complete path} is a path starting from the source vertex and ending at the sink vertex.
Formally, if a path $(\pv_0,\cdots,\pv_k)$ satisfies $\pv_0 = v_{src}$ and $\pv_k = v_{snk}$, then it is a complete path.
The \emph{longest path} is a complete path with the largest length among all paths in $G$.
The length of the longest path in $G$ is denoted as $len(G)$.
For a vertex set $V'\subseteq V$, we define $vol(V') \coloneqq \sum_{v\in V'}c(v)$.
The volume of $G$ is the total workload in $G$ and is defined as $vol(G) \coloneqq \sum_{v\in V} c(v)$.
If there is an edge $(u,v)$, we say that $u$ is a \emph{predecessor} of $v$, and $v$ is a \emph{successor} of $u$.
If there is a path starting from $u$ and ending at $v$, we say that $u$ is an \emph{ancestor} of $v$ and $v$ is a \emph{descendant} of $u$.
The sets of predecessors, successors, ancestors and descendants of $v$ are denoted as $pre(v)$, $suc(v)$, $anc(v)$ and $des(v)$, respectively.

A \emph{generalized path} $\chain=(\pv_0, \cdots, \pv_k)$ is a set of vertices such that $\forall i\in [0,k-1]$: $\pv_i$ is an ancestor of $\pv_{i+1}$. In particular, a vertex set containing only one vertex is a generalized path.
By definition, a path is a generalized path, but a generalized path is not necessarily a path.
Similar to paths, the length of a generalized path $\chain$ is defined as $len(\chain) \coloneqq \sum_{\pv_i\in \chain} c(\pv_i)$.

\begin{example}\label{exp:dag_example}
Fig. \ref{fig:dag_example} shows a DAG task $G$. The number inside vertices is the WCET of this vertex.
The source vertex and the sink vertex are $v_0$ and $v_5$, respectively.
For vertex set $V'=\{v_1, v_2\}$, $vol(V')=4$. The volume of the DAG task is $vol(G)=10$.
For vertex $v_4$, $pre(v_4)=\{v_1, v_2\}$, $suc(v_4)=\{v_5\}$, $anc(v_4)=\{v_0, v_1, v_2\}$, $des(v_4)=\{v_5\}$.
The longest path is $\mypath=(v_0, v_1, v_4, v_5)$, and $len(G) = len(\mypath)=6$.
$\chain=(v_0, v_2, v_5)$ is a generalized path. Note that by definition $\chain$ is not a path, because $(v_2, v_5)$ is not an edge in the DAG task.
\end{example}

\subsection{Scheduling Model}
\label{sec:schedule}
We consider that the DAG task $G$ executes on a computing platform with $m$ identical cores.
A vertex $v$ is said to be \emph{eligible} if all of its predecessors have completed their execution, thus $v$ can be immediately executed if there are available cores.
For a scheduling algorithm, the \emph{work-conserving} property means that an eligible vertex must be executed if there are available cores.
We do not restrict the scheduling algorithm, as long as it satisfies the work-conserving property.
Without loss of generality, we assume the source vertex of $G$ starts its execution at time $0$.
The \emph{response time} of $G$ is defined as the time when the sink vertex finishes its execution.
For example, a possible schedule of the DAG task $G$ in Fig. \ref{fig:dag_example} under a work-conserving scheduler is shown in Fig. \ref{fig:sequence1}. The response time of $G$ in this schedule is 7.

The problem model presented above assumes that the system contains only one DAG task.
Later in Section \ref{sec:extension}, we will extend our developed techniques to deal with systems consisting of multiple DAG tasks.

\section{Motivation}
\label{sec:motivation}
This section uses two examples to provide the intuition that motivates this work.
The first example is relatively simple and illustrates the effect of adding edges.
The second example is more technical and explains the challenge of adding edges.

\subsection{The First Example}
\label{sec:motivation_first}

\begin{figure}[t]
\centering
\begin{minipage}[c]{0.23\textwidth}
\qquad
\subfloat[the DAG task model]{   
    \includegraphics[width=0.7\linewidth]{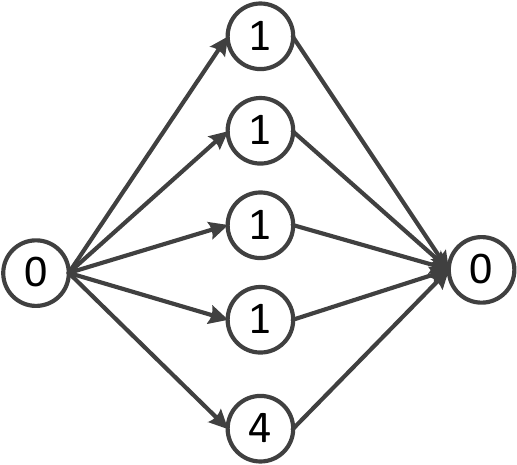}
    \label{fig:dag_intro1} 
}
\end{minipage}
\hfill
\begin{minipage}[c]{0.25\textwidth}
\subfloat[a possible schedule of Fig. \ref{fig:dag_intro1}]{
    \quad
    \includegraphics[width=0.69\linewidth]{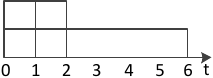}
    \label{fig:seq_intro1}
}\
\subfloat[a schedule after adding edges]{
    \qquad
    \includegraphics[width=0.53\linewidth]{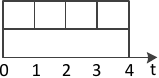}
    \label{fig:seq_intro2}
    \quad
}
\end{minipage}
\caption{The first motivational example.}
\label{fig:intro}
\end{figure}

This subsection uses an example from autonomous driving systems to illustrate the effect of adding edges.
The system has one front camera to monitor the road ahead and four LiDARs for proximity sensing.
There are two ECUs (electronic control unit) for processing the collected information.
Since deep neural networks, which are time-consuming, are usually utilized for processing the scenes captured by the camera, we suppose the execution time of the camera task is 4 and the execution time of each LiDAR task is 1.
This system can be modeled as a DAG task shown in Fig. \ref{fig:dag_intro1}. Note that by Section \ref{sec:task}, two dummy vertices with zero WCET are added to produce a DAG with single source and single sink.
For the DAG task in Fig. \ref{fig:dag_intro1}, a possible schedule is shown in Fig. \ref{fig:seq_intro1}, whose response time is 6.
If we add three edges among the four LiDAR tasks, which ensures that the four LiDAR tasks can only be executed sequentially, the schedule in Fig. \ref{fig:seq_intro1} will be impossible and Fig. \ref{fig:seq_intro2} shows a possible schedule.
In Fig.~\ref{fig:seq_intro2}, the response time is 4, better than that of Fig. \ref{fig:seq_intro1}.

In Fig. \ref{fig:seq_intro1}, it can be easily seen that the bottleneck that affects the response time is the longest path (i.e., the camera task with WCET of 4). And the four LiDAR tasks interfere with the execution of the camera task, which leads to a response time of 6.
However, in Fig. \ref{fig:seq_intro2}, by adding edges, the interference to the longest path is eliminated, which leads to a response time of 4.
This example clearly demonstrates how the worst-case interference is reduced and the worst-case response time is shortened by properly adding edges and constructing longer paths.
We remark that although adding edges is valid, removing edges is not, since edges in DAG tasks usually mean data dependencies or synchronization among vertices.

\subsection{The Second Example}
\label{sec:motivation_second}

The example in Section \ref{sec:motivation_first} demonstrates that strengthening precedence constraints by adding edges has the potential to reduce the response time of a DAG task.
However, to explore this potential to improve the worst-case response time guarantee of a DAG task, we need to solve two problems.
First, systematically designing methods of adding edges requires the guidance of related theories. 
Second, real-time scheduling requires theoretical guarantees concerning the response time of tasks.
To address these two challenges, we employ the state-of-the-art analysis technique in \cite{he2022bounding}, which utilizes the information of multiple long paths in the DAG task to shorten the response time bound.
Next, we briefly introduce this technique.

\begin{definition}[Generalized Path List]\label{def:chain_list}
A generalized path list is a set of disjoint generalized paths $(\chain_i)_0^k$ ($k \ge 0$), i.e.,
$$\forall i, j \in [0, k],\ \chain_i \cap \chain_j = \varnothing$$
\end{definition}
Here $(\chain_i)_0^k$ is the compact representation of $(\chain_0, \cdots, \chain_k)$.

\begin{theorem}[\cite{he2022bounding}]\label{thm:he_bound}
Given a generalized path list $(\chain_i)_0^k$ ($k \in [0, m-1]$) where $\chain_0$ is the longest path of $G$, the response time of DAG task $G$ scheduled by a work-conserving scheduler on $m$ cores is bounded by $R(G)$.
\begin{equation}\label{equ:he_bound}
R(G) \coloneqq \min \limits_{j \in [0, k]} \left\{ len(G)+\frac{vol(G)-\sum_{i=0}^{j} len(\chain_i)}{m-j} \right\}
\end{equation}
\end{theorem}

\begin{example}\label{exp:he_bound}
For the DAG $G$ in Fig. \ref{fig:dag_example}, we can identify a generalized path list $(\chain_i)_0^2$ where $\chain_0=(v_0, v_1, v_4, v_5)$, $\chain_1=(v_3)$, $\chain_2=(v_2)$.
Let the number of cores $m=2$. For this generalized path list, the bound in (\ref{thm:he_bound}) is
$R(G)=\min\{6+(10-6)/2, 6+(10-6-3)/(2-1)\}=\min\{8, 7\}=7$.
\end{example}

\begin{figure}[t]
\centering
\subfloat[the DAG with edge $(v_2, v_3)$]{
    \includegraphics[width=0.4\linewidth]{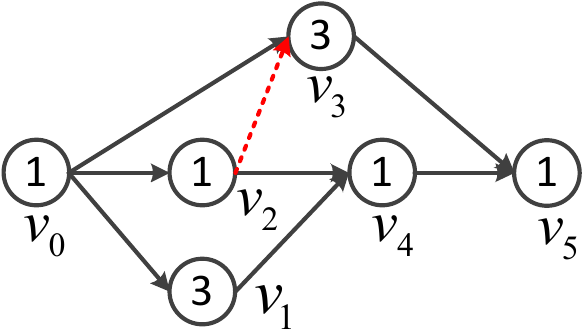}
    \label{fig:dag_example2}
}
\hfil
\subfloat[a possible schedule]{
    \includegraphics[width=0.4\linewidth]{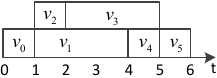}
    \label{fig:sequence2}
}
\caption{The second motivational example.}
\label{fig:motivation}
\end{figure}

Next, we use Fig. \ref{fig:motivation} to explain how adding edges can reduce the response time bound.
The DAG $G'$ in Fig. \ref{fig:dag_example2} is obtained by adding edge $(v_2, v_3)$ to the DAG task $G$ in Fig. \ref{fig:dag_example}.
After adding the edge, the schedule in Fig. \ref{fig:sequence1} becomes invalid and a new possible schedule is shown in Fig. \ref{fig:sequence2}.
In $G'$, we can find a generalized path list $\chain_0=(v_0, v_1, v_4, v_5)$, $\chain_1=(v_2, v_3)$.
And the bound in (\ref{thm:he_bound}) is computed as
$R(G')=\min\{6+(10-6)/2, 6+(10-6-4)/(2-1)\}=\min\{8, 6\}=6$,
which is smaller than $R(G)=7$ in Example \ref{exp:he_bound}.
It can be seen from the examples that there is a strong connection between the response time and the response time bound: for the DAG $G$ in Fig. \ref{fig:dag_example}, the response time is 7 and the bound is 7; after adding edges, for the DAG $G'$ in Fig. \ref{fig:dag_example2}, the response time is 6 and the bound is 6.

Note that adding edges in DAG tasks does not necessarily lead to a smaller response time or smaller response time bound.
We still take the DAG $G$ in Fig. \ref{fig:dag_example} to illustrate this. Let the number of cores $m=2$.
If we add an edge $(v_1, v_3)$, which transforms $G$ into another DAG $G''$, it can be easily seen that the response time of $G''$ can be 8.
In $G''$, we can find a generalized path list $\chain_0=(v_0, v_1, v_3, v_5)$, $\chain_1=(v_2, v_4)$.
And the bound in (\ref{thm:he_bound}) is computed as
$R(G'')=\min\{8+(10-8)/2, 8+(10-8-2)/(2-1)\}=\min\{9, 8\}=8$,
which is larger than $R(G)=7$ in Fig. \ref{fig:dag_example} and $R(G')=6$ in Fig. \ref{fig:dag_example2}.

Therefore, the policies for adding edges should be delicately developed.
In this paper, we carefully examine the dependencies between vertices in DAG tasks and design techniques of adding edges to reduce the response time bound, thus generally reducing the response time of DAG tasks.

\section{Making Long Paths}
\label{sec:method}

This section presents the method of adding edges to make long paths through examining the dependencies among vertices of the DAG task.
In the following, we first discuss the principles of adding edges, then present the algorithm of adding edges to reduce the response time bound of DAG task $G$.

\begin{definition}[Parallel Set] \label{def:parallel_set}
The parallel set of vertex $v$ is defined as
$para(v) \coloneqq \{u\in V \setminus \{v\}|u\notin anc(v) \land u\notin des(v) \}$.
\end{definition}
The parallel set of vertex $v$ contains vertices that can execute in parallel with $v$.
For example, in Fig. \ref{fig:dag_example}, $para(v_1)=\{v_2, v_3\}$.

Concerning adding edge $(u, v)$ for distinct vertices $u$, $v$ (i.e., $u \ne v$), if $u$ is an ancestor of $v$, since there is already precedence relation between $u$ and $v$, adding edge does not change the execution behavior of the DAG task. Therefore, if $u$ is an ancestor of $v$, there is no need for adding edge $(u, v)$. If $u$ is a descendant of $v$, adding edge $(u, v)$ leads to circles in the DAG, which violates the semantics of directed acyclic graph. Therefore, if $u$ is a descendant of $v$, we cannot add edge $(u, v)$. We summarize this into Principle \ref{prp:para}.

\begin{principle}\label{prp:para}
For distinct vertices $u$, $v$, if $u \in para(v)$, then edge $(u, v)$ can be added in the DAG task.
\end{principle}

Principle \ref{prp:para} specifies whether an edge can be added or not, but does not indicate whether an edge should be added or not. Recall that our target of adding edge is to reduce the response time bound in (\ref{thm:he_bound}). If the length of the longest path increases after adding edges, i.e., $len(G)$ becomes larger, we cannot guarantee that the response time bound in (\ref{thm:he_bound}) be reduced. We summarize this into Principle \ref{prp:longest}.

\begin{principle}\label{prp:longest}
For distinct vertices $u$, $v$, if the length of the longest path does not increase after adding edge $(u, v)$, then edge $(u, v)$ can be added in the DAG task.
\end{principle}
To realize Principle \ref{prp:longest} and make it easier to implement, we introduce the following concepts.

\begin{definition}[Left Length] \label{def:left}
In DAG $G$, the left length $l(v)$ of vertex $v$ is the maximum length of paths that start from $v_{src}$ and end at $v$.
\end{definition}
\begin{definition}[Right Length] \label{def:right}
In DAG $G$, the right length $r(v)$ of vertex $v$ is the maximum length of paths that start from $v$ and end at $v_{snk}$.
\end{definition}

The left length and right length of vertices in a DAG can be easily computed by using a simple dynamic programming with time complexity $O(|V|+|E|)$ \cite{he2019intra}.
For example, in Fig. \ref{fig:dag_example}, for vertex $v_4$, $l(v_4)=5$, $r(v_4)=2$.

\begin{lemma}\label{lem:longest}
In DAG $G$, for vertex $u$, $v$ and $u \in para(v)$, we add edge $(u, v)$ which transforms $G$ to $G'$. If in $G$, $l(u)+r(v) \le len(G)$, then $len(G')=len(G)$.
\end{lemma}
\begin{proof}
Let $\Pi$ and $\Pi'$ denote the set of complete paths in $G$ and $G'$, respectively. $\Pi'$ can be partitioned into two sets $\Pi'_1$, $\Pi'_2$ where $\Pi'_1$ is the set of complete paths that do not go through $(u, v)$, and $\Pi'_2$ is the set of complete paths that go through $(u, v)$. Obviously, $\Pi'_1= \Pi$, which means
\begin{equation}\label{equ:longest1}
\max_{\mypath \in \Pi'_1} \{len(\mypath)\} = \max_{\mypath \in \Pi} \{len(\mypath)\} = len(G)
\end{equation}
Since $l(u)+r(v) \le len(G)$, we have
\begin{equation}\label{equ:longest2}
\max_{\mypath \in \Pi'_2} \{len(\mypath)\} \le len(G)
\end{equation}
Combining (\ref{equ:longest1}) and (\ref{equ:longest2}), we reach the conclusion.
\end{proof}

Lemma \ref{lem:longest} will be used to implement Principle \ref{prp:longest} in Algorithm \ref{alg:first}.
Principle \ref{prp:longest} is to ensure that the response time bound in (\ref{thm:he_bound}) does not become larger. However, we want the response time bound to become smaller. Next, we introduce Principle \ref{prp:effective}.

\begin{definition}[Residue Graph \cite{he2022bounding}]\label{def:residue_graph}
Given a generalized path $\chain$ of DAG $G = (V, E)$, the residue graph $G_r = res(G, \chain)= (V, E)$ is defined as:
\begin{itemize}
  \item if $v \in \chain$, the WCET of $v$ in $G_r$ is $0$;
  \item if $v \in V \setminus \chain$, the WCET of $v$ in $G_r$ is $c(v)$.
\end{itemize}
\end{definition}

Note that different generalized paths in a generalized path list are disjoint. The residue graph is introduced to ensure that newly computed generalized paths have no common vertices with already computed generalized paths while maintaining the topology of the DAG task.
A residue graph is still a DAG.

\begin{principle}\label{prp:effective}
For distinct vertices $u$, $v$, if a longer generalized path can be identified in the residue graph after adding edge $(u, v)$, then edge $(u, v)$ can be added in the DAG task.
\end{principle}

Next, we introduce some concepts to realize Principle \ref{prp:effective}.
For a vertex $v$, the \emph{effective left length} $el(v)$ is the left length of $v$ in a residue graph $G_r$,
and the \emph{effective right length} $er(v)$ is the right length of $v$ in a residue graph $G_r$.

\begin{example}\label{exp:effective}
For the DAG $G$ in Fig. \ref{fig:dag_example} and a generalized path $\chain = (v_0, v_3, v_5)$, the residue graph $G_r= res(G, \chain)$ where the WCETs of $v_0$, $v_3$, $v_5$ are set to zero; edges and other vertices are unchanged.
In $G_r$, for vertex $v_4$, $el(v_4)=4$, $er(v_4)=1$.
\end{example}

\begin{lemma}\label{lem:effective}
In residue graph $G_r$, for vertex $u$, $v$ and $u \in para(v)$, we add edge $(u, v)$ which transforms $G_r$ to $G'_r$. If in $G_r$, $el(u)+er(v) > len(G_r)$, then $len(G'_r)>len(G_r)$.
\end{lemma}
\begin{proof}
In $G_r$, let the path that starts from $v_{src}$ and ends at $u$ and whose length reaches $el(u)$ be $\mypath_1$. Let the path that starts from $v$ and ends at $v_{snk}$ and whose length reaches $er(v)$ be $\mypath_2$. Since edge $(u, v)$ is in $G'_r$, we have $\mypath \coloneqq \mypath_1 \cup \mypath_2$ is a path in $G'_r$ and $len(\mypath)=el(u)+er(v) > len(G_r)$.
Therefore, $len(G'_r)>len(G_r)$.
\end{proof}

Lemma \ref{lem:effective} will be used to implement Principle \ref{prp:effective} in Algorithm \ref{alg:first}.
With these three principles, we present our method of adding edges in Algorithm \ref{alg:framework} and Algorithm \ref{alg:first}.

\begin{algorithm}[t]
    \caption{The Method Framework}\label{alg:framework}
    \DontPrintSemicolon
    \Input{DAG task $G = (V, E)$}
    \Output{DAG task $G$ with added edges}
     $G_r \leftarrow G$; $i \leftarrow 0$ \\
     \While{$vol(G_r) \neq 0$}{
        $\chain_i \leftarrow$ the longest path of $G_r$ \\
        \If{$AddEdge(\chain_i)$}{
            \Continue
        }
        $\chain_i \leftarrow \chain_i \setminus \{v \in \chain_i| c(v)\ \mathrm{of}\ G_r\ \mathrm{is}\ 0 \}$ \\
        $G_r \leftarrow res(G_r, \chain_i)$; $i \leftarrow i+1$ \\
     }
\end{algorithm}

\begin{algorithm}[t]
    \caption{$AddEdge(\chain)$}\label{alg:first}
    \DontPrintSemicolon
    \Input{a generalized path $\chain$}
    \Output{\emph{true}: add an edge; \emph{false}: not add an edge}
    \ForEach{$v \in \chain$}{
        \ForEach(\tcp*[f]{Principle \ref{prp:para}}){$u \in para(v)$}{
            \If(\tcp*[f]{Principle \ref{prp:longest}}){$l(u)+r(v) \le len(G)$}{
                \If(\tcp*[f]{Principle \ref{prp:effective}}){$el(u)+el(v) > len(G_r)$}{
                    add edge $(u, v)$ in $G$ and $G_r$ \\
                    \Return{true}
                }
            }
        }
    }
    \Return{false}
\end{algorithm}

Algorithm \ref{alg:framework} follows the same guideline as in Algorithm 2 of \cite{he2022bounding}: it computes the longest generalized path in the residue graph one by one (Line 3), and sets the WCETs of vertices in these generalized paths to zero to avoid joint vertices among different generalized paths (Line 7, 8) until the volume of the residue graph is zero (Line 2).
Whenever a longest generalized path is computed in Line 3, the \emph{AddEdge} procedure is called regarding this generalized path (Line 4). If an edge is successfully added, the topology of the DAG task changes and the longest generalized path should be recomputed (Line 3).

The \emph{AddEdge} procedure (Algorithm \ref{alg:first}) performs the task of adding edge and realizes the three proposed principles. It takes a generalized path $\chain$ as the input. For each vertex $v$ in this generalized path (Line 1), for each vertex $u$ in the parallel set of $v$ (Line 2, this line corresponds to Principle \ref{prp:para}), we check whether edge $(u, v)$ can be added or not. We first check whether it conforms to Principle \ref{prp:longest} using Lemma \ref{lem:longest} (Line 3), second check Principle \ref{prp:effective} using Lemma \ref{lem:effective} (Line 4). If edge $(u, v)$ passes these checks, we add this edge in the DAG task (Line 5).
If we cannot make generalized path $\chain$ longer by adding edges, which means that it is a ``good'' generalized path to compute the response time bound for the DAG task, we indicate that there is no need to add edges and return \emph{false} (Line 11).


\textbf{Complexity.}
There are two loops in Algorithm \ref{alg:first} and each of them can run at most $|V|$ times. Therefore, the time complexity of Algorithm \ref{alg:first} is $O(|V|^2)$.
For Algorithm \ref{alg:framework}, there are two independent loops: the loop in Line 2-5 and the loop in Line 2-9.
For each iteration of the first loop, an edge is added. Since a DAG can have at most $|V|^2$ edges, the first loop can run no more than $|V|^2$ times.
For each iteration of the second loop, a generalized path is identified, which includes at least one vertex. The second loop can run no more than $|V|$ times.
Together with the $O(|V|^2)$ of Algorithm \ref{alg:first}, the time complexity of Algorithm \ref{alg:framework} is $O(|V|^4)$.

For the original DAG task $G$, after Algorithm \ref{alg:framework}, a DAG with added edges (denoted as $G'$) is computed.
Also, in Algorithm \ref{alg:framework}, during each iteration, a generalized path $\chain_i$ is computed in Line 7. These generalized paths form a generalized path list $(\chain_i)_0^k$, which can be used to compute the response time bound for DAG $G'$ using (\ref{equ:he_bound}).
Note that in Theorem \ref{thm:he_bound}, the sole requirement that a generalized path list can be used to compute the bound in (\ref{equ:he_bound}) is that the first generalized path is the longest path of the DAG task.

\begin{theorem}\label{thm:domination}
The method in Algorithm \ref{alg:framework} and Algorithm \ref{alg:first} dominates the method in \cite{he2022bounding}, i.e.,
\begin{equation}\label{equ:domination}
R(G') \le R(G)
\end{equation}
\end{theorem}
\begin{proof}
We prove it by examining the items of (\ref{equ:he_bound}).
Since the WCETs of vertices in $G'$ and $G$ are the same, $vol(G')=vol(G)$.
Although we add edges in $G'$, by Lemma \ref{lem:longest}, $len(G')=len(G)$.
For a generalized path list $(\chain_i)_0^k$ of $G$ to compute $R(G)$, since $\chain_0$ is the longest path of $G$ and $len(G')=len(G)$, $\chain_0$ is also the longest path of $G'$, which means that $(\chain_i)_0^k$ can also be used to compute $R(G')$.
Therefore, $R(G')$ cannot be larger than $R(G)$.
What's more, by Lemma \ref{lem:effective}, after adding an edge, it is ensured that there is a longer generalized path in the residue graph (see Line 4 of Algorithm \ref{alg:first}).
In summary, the conclusion is reached.
\end{proof}

\begin{example}\label{exp:perserve}
Let the number of cores $m=2$.
For the DAG task $G$ in Fig. \ref{fig:dag_example}, the computed bound is 7. After Algorithm \ref{alg:framework} and Algorithm \ref{alg:first}, $G$ is transformed into $G'$ shown in Fig. \ref{fig:dag_example2}. The computed bound is 6, less than 7. This example is explained in Section \ref{sec:motivation_second}.
\end{example}

\section{Real-Time Scheduling of DAG Tasks}
\label{sec:extension}
This section considers the scheduling of a task set.
In the task set, each sporadic parallel real-time task is specified as a tuple $(G, D, T)$, where $G$ is the DAG task model in Section \ref{sec:task}, $D$ is the relative deadline and $T$ is the period. We consider constrained deadline, i.e., $D \le T$.
The scheduling algorithm is the widely-used federated scheduling paradigm \cite{li2014analysis}.
In federated scheduling paradigm, parallel tasks are divided into two categories: the heavy tasks (tasks with $vol(G) \ge D$) and the light tasks (tasks with $vol(G) < D$).
Each heavy task is assigned and executed exclusively on a set of cores under a work-conserving scheduler.
All light tasks are treated as sequential sporadic tasks and are scheduled on the remaining cores by sequential multiprocessor scheduling algorithms such as global EDF \cite{baruah2007techniques} or partitioned EDF \cite{baruah2005partitioned}.

In federated scheduling paradigm, to apply the proposed technique, the only remaining question is how to decide the number of cores $m$ required by a heavy task such that its deadline can be satisfied.
In the following, Section \ref{sec:increase} develops a technique to optimize the required number of cores for a DAG task. And Section \ref{sec:task_set} presents the scheduling approach using the techniques in Section \ref{sec:method} and Section \ref{sec:increase}.

\subsection{Optimizing the Number of Cores}
\label{sec:increase}  

Section \ref{sec:method} considers how to optimize the response time bound for a DAG task given the number of cores.
In contrast, Section \ref{sec:extension} considers computing resource allocation. So the target of this subsection is to optimize the number of cores, instead of the response time bound.

Section \ref{sec:method} enforces that the longest path cannot be increased (i.e., Principle \ref{prp:longest}) when adding edges. Given the number of cores, to reduce the response time bound, this is a reasonable principle, since the length of the longest path has a large impact on the worst-case response time.
However, given the deadline, to reduce the number of cores required by a DAG task, it is possible that lifting the constraint of Principle \ref{prp:longest} can achieve better computing resource allocations.

\textbf{Motivation.}
We illustrate this using the DAG task $G$ in Fig. \ref{fig:dag_example3}. Let the deadline $D=7$.
In $G$, a generalized path list $(\chain_i)_0^2$ where $\chain_0=(v_0, v_1, v_4, v_5)$, $\chain_1=(v_3)$, $\chain_2=(v_2)$ can be identified. Since there are three generalized paths in $G$, if the allocated number of cores $m=3$, the worst-case response time will be the length of the longest path in $G$, which is 6.
If $m=2$, using this generalized path list, the response time bound of $G$ in (\ref{equ:he_bound}) can be computed as
$R(G) =\min\{6+(11-6)/2, 6+(11-6-3)/(2-1)\}=\min\{8.5, 8\}=8$,
which is larger than the deadline.
Also, due to Principle \ref{prp:longest}, no edges can be added using the method in Section \ref{sec:method}.
Therefore, to guarantee that the deadline is satisfied, the allocated number of cores $m$ should be 3.

However, if without Principle \ref{prp:longest}, we can add edge $(v_2, v_3)$ and transforms $G$ into $G'$ shown in Fig. \ref{fig:dag_example4}. In $G'$, a generalized path list $(\chain_i)_0^1$ where $\chain_0=(v_0, v_2, v_3, v_5)$, $\chain_1=(v_1, v_4)$ can be identified.  Since there are two generalized paths in $G'$, if the allocated number of cores $m=2$, the worst-case response time will be the length of the longest path in $G'$, which is 7, no larger than the deadline.
Therefore, if without Principle \ref{prp:longest}, to guarantee that the deadline is satisfied, the allocated number of cores $m$ can be 2.
In this example, by lifting the constraint of Principle \ref{prp:longest}, the allocated number of cores is reduced without compromising the hard real-time requirements.

\begin{figure}[t]
\centering
\subfloat[DAG task $G$]{
    \includegraphics[width=0.4\linewidth]{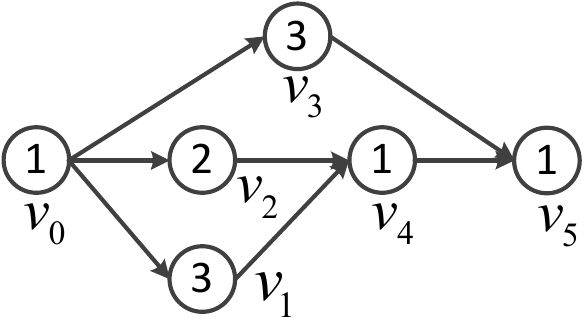}
    \label{fig:dag_example3}
}
\hfil
\subfloat[DAG $G'$ with edge $(v_2, v_3)$]{
    \quad
    \includegraphics[width=0.4\linewidth]{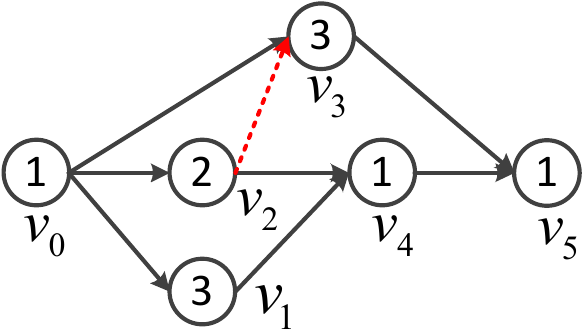}
    \label{fig:dag_example4}
    \quad
}
\caption{An illustrative example for Section \ref{sec:increase}.}
\label{fig:increase}
\end{figure}

\textbf{The Proposed Technique.} 
Therefore, to reduce the allocated number of cores, the technique of adding edges without Principle \ref{prp:longest} is proposed.
We introduce a variable called \emph{limit}, which is an upper bound of the length of the longest path when adding edges.
The condition of the \emph{if} statement in Line 3 of Algorithm \ref{alg:first} is replaced with
\begin{equation}\label{equ:increase}
l(u)+r(v) \le \emph{limit}
\end{equation}
Others in Algorithm \ref{alg:framework} and Algorithm \ref{alg:first} are unchanged.
It can be seen that the meaningful range of \emph{limit} is $[len(G), vol(G)]$. If $\emph{limit}=len(G)$, this technique is the method in Section \ref{sec:method}; if $\emph{limit}=vol(G)$, the DAG task will be transformed into a sequential task.

\textbf{The Proposed Technique with Given Deadline.} 
If the deadline $D$ of the DAG task is provided, since the only requirement of real-time scheduling is to ensure that the task can finish before its deadline, we let the condition of the \emph{if} statement in Line 3 of Algorithm \ref{alg:first} be
\begin{equation}\label{equ:deadline}
l(u)+r(v) \le D
\end{equation}
Others in Algorithm \ref{alg:framework} and Algorithm \ref{alg:first} are unchanged.
Recall that during each iteration of Algorithm \ref{alg:framework}, a generalized path $\chain_i$ is computed in Line 7.
So Algorithm \ref{alg:framework}, with Line 3 modified according to (\ref{equ:deadline}), will also output a generalized path list $(\chain_i)_0^k$.

\begin{lemma}\label{lem:increase}
Let $(\chain_i)_0^k$ denote the generalized path list computed by Algorithm \ref{alg:framework} with (\ref{equ:deadline}). If the allocated number of cores $m=k+1$, the DAG task will finish before its deadline.
\end{lemma}
\begin{proof}
Let $G$ denote the DAG task, and $G'$ denote the DAG after Algorithm \ref{alg:framework} with Line 3 modified according to (\ref{equ:deadline}).
Since adding edges does not change the volume, we have $vol(G)=vol(G')$.
By Line 2 of Algorithm \ref{alg:framework}, since the loop iterates until the volume of the residue graph reaches zero, we have $vol(G)=\sum_{i=0}^{k} len(\chain_i)$.
By (\ref{equ:he_bound}),
\begin{align*}
R(G') &= \min \limits_{j \in [0, k]} \left\{ len(G')+\frac{vol(G')-\sum_{i=0}^{j} len(\chain_i)}{m-j} \right\} \\
      &= len(G')+\frac{vol(G)-\sum_{i=0}^{k} len(\chain_i)}{m-k} = len(G')
\end{align*}
By (\ref{equ:deadline}) and the proof of Lemma \ref{lem:longest}, we have $len(G') \le D$.
Therefore, $R(G') \le D$.
\end{proof}

\subsection{The Scheduling Approach}
\label{sec:task_set}
This subsection discusses in the proposed scheduling approach, how to compute the number of cores allocated to a DAG task.
For a DAG task $G$, \cite{he2022bounding} presented a method to compute the allocated number of cores such that the deadline is guaranteed (Theorem 3 of \cite{he2022bounding}).
Let $Path(G)$ denote this method in \cite{he2022bounding}, which takes a DAG $G$ as input and outputs the number of cores $m$.
The technique in Section \ref{sec:method} transforms the DAG task $G$ into another DAG $G'$.
Using the method of \cite{he2022bounding}, we can compute a valid number of cores $Path(G')$.

The technique in Section \ref{sec:increase} also transforms the DAG task $G$ into another DAG $G''$ and computes the allocated number of cores such that the deadline is guaranteed (see Lemma \ref{lem:increase}).
Let $Edge(G)$ denote this method, which takes a DAG $G$ as input and outputs the number of cores $m$.

In the proposed scheduling approach, for a DAG task $G$, the allocated number of cores $m$ is computed as
\begin{equation}\label{equ:task_set}
m= \min \{Path(G'), Edge(G)\}
\end{equation}
Since for a DAG task, both allocated numbers of cores (i.e., $Path(G')$, $Edge(G)$) can guarantee that the deadline is satisfied, we can safely use the smaller one as the final number of cores allocated to the DAG task.

\begin{theorem}\label{thm:task_set}
For scheduling a task set, the approach in Section \ref{sec:extension} dominates the approach in \cite{he2022bounding} in the sense that if \cite{he2022bounding} can schedule a task set, our approach can schedule this task set.
\end{theorem}
\begin{proof}
For an arbitrary task $G$ in this task set, let $G'$ denote the DAG produced by the technique in Section \ref{sec:method}.
By Theorem \ref{thm:domination}, $R(G') \le R(G)$, which means that for a deadline $D$, the number of cores required by $G'$ is no larger than that of $G$.
Therefore, $Path(G') \le Path(G)$, which means that the computed number of cores $m$ in (\ref{equ:task_set}) is no larger than $Path(G)$.
The conclusion is reached.
\end{proof}

Note that both the proposed scheduling approach and the original federated scheduling \cite{li2014analysis} belong to the federated scheduling paradigm, where heavy tasks are allocated and executed exclusively on a set of cores.
The critical difference among various federated scheduling approaches lies in the method of computing resource allocation.
The only difference between our approach and the approach in \cite{li2014analysis} also lies in resource allocation, i.e., the number of cores allocated to heavy tasks.
Compared to federated scheduling approaches such as \cite{li2014analysis, he2022bounding}, benefiting from the proposed techniques that reduce the number of cores for heavy tasks, our scheduling approach significantly improves the system schedulability (see Section \ref{sec:evaluation_set}).
Compared to federated scheduling approaches such as \cite{jiang2017semi, ueter2018reservation, jiang2021virtually}, where heavy tasks can share computing resources to some extent, not only the performance of the proposed approach is better, but also the implementation of the proposed approach is much easier, since in our approach there are not sophisticated policies for heavy tasks to share computing resources.

\section{Evaluation}
\label{sec:evaluation}
This section evaluates the performances of the proposed methods for scheduling one task and scheduling task sets.

\subsection{Evaluation of Scheduling One Task}
\label{sec:evaluation_one}

This subsection evaluates the response time bounds of one DAG task using the following methods.
\begin{itemize}
    \item \textsf{PATH}. The method in \cite{he2022bounding}, shown in Theorem \ref{thm:he_bound}.
    \item \textsf{OUR}. Our method presented in Section \ref{sec:method}.
\end{itemize}
The bound in \cite{he2022bounding} is the state-of-the-art regarding scheduling a DAG task under a work-conserving scheduler on an identical multi-core platform.
Both bounds are normalized with respect to Graham's bound in \cite{graham1969bounds} to compare the performances.

\textbf{Task Generation.}
The DAG tasks are generated using the Erd\"os-R\'enyi method \cite{cordeiro2010random}. The number of vertices $|V|$ is randomly chosen in a specified range. For each pair of vertices, it generates a random value in $[0, 1]$ and adds an edge to the graph if the generated value is less than a predefined \emph{parallelism factor} $\mathit{pf}$. The larger $\mathit{pf}$, which means that there are more edges, the more sequential the graph is. After the vertices and edges are generated, the WCET of each vertex is randomly chosen in a specified range.
The default settings for generating DAG tasks are as follows.
The WCETs of vertices $c(v)$, the vertex number $|V|$, and the parallelism factor $\mathit{pf}$ are randomly and uniformly chosen in $[50, 100]$, $[50, 250]$ and $[0, 0.5]$, respectively.
For each data point in Fig. \ref{fig:evaluation_one}, we randomly generate 500 tasks to compute the average normalized bound.

The experiment results are reported in Fig. \ref{fig:evaluation_one}.
Fig. \ref{fig:m} shows the results of changing the number of cores on which the DAG task is scheduled.
When $m=2$, since fewer vertices can execute in parallel, the execution of DAG tasks is more sequential.
Therefore, both bounds are close to the volume of the task. This is the reason why the data points of both bounds in Fig. \ref{fig:m} are relatively close to each other when $m=2$.
When $m$ gets larger gradually (i.e., $m \in [3, 7]$), the technique of adding edges, being able to construct longer paths and reduce the interference to the critical path, becomes more effective in reducing the bound.
This explains why the data points of \textsf{OUR} in Fig. \ref{fig:m} become decrease when $m \le 4$.
When $m$ is close to 10, since more vertices can execute immediately once being released, the response time will approach the length of the longest path. Therefore, both bounds are close to the length of the longest path.
This explains why the data points of \textsf{OUR} in Fig. \ref{fig:m} become increase when $m \ge 4$.
Compared to \textsf{PATH}, our method can reduce the normalized bound by 21.6\% when $m=4$.
We choose $m=4$ as a representative for the following two experiments.

\begin{figure}[t]
\centering
\subfloat[core number]{
    \includegraphics[width=0.49\linewidth]{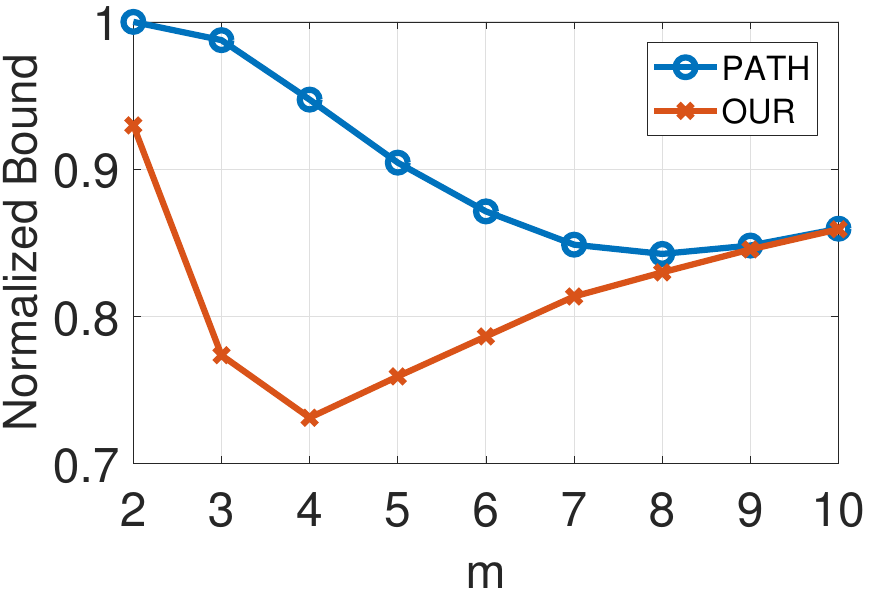}
    \label{fig:m}
}
\subfloat[parallelism factor]{
    \includegraphics[width=0.49\linewidth]{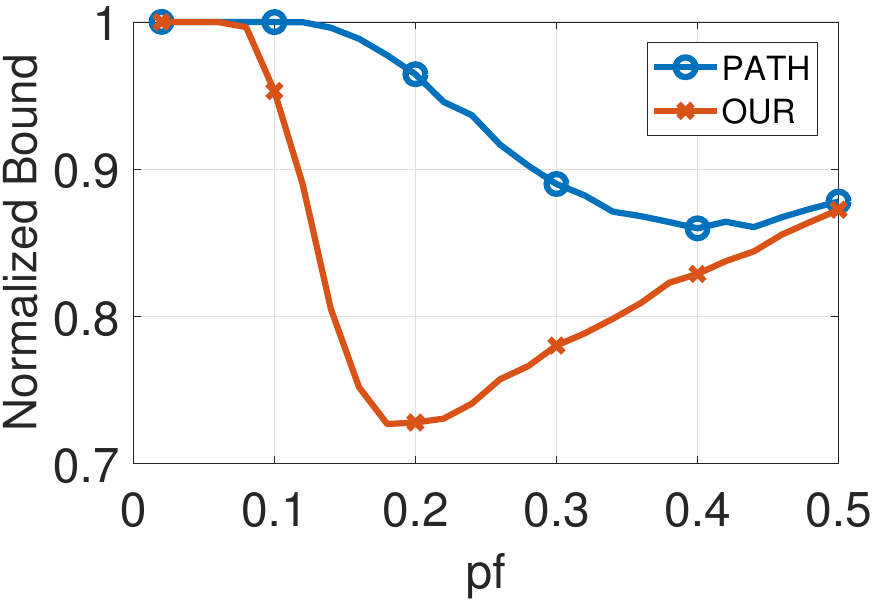}
    \label{fig:pf}
}
\hfil
\subfloat[vertex number]{
    \includegraphics[width=0.49\linewidth]{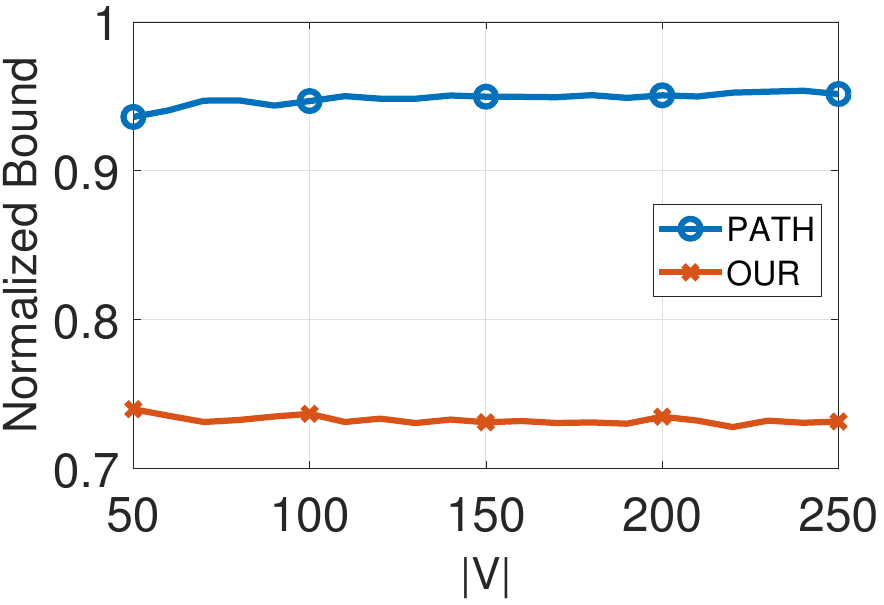}
    \label{fig:v}
}
\caption{Evaluation results of one task.}
\label{fig:evaluation_one}
\end{figure}

Fig. \ref{fig:pf} reports the results of changing the parallelism factor $\mathit{pf}$.
When $\mathit{pf}$ is close to 0, since there are fewer edges in the graph, all paths (including the longest path) are short.
For \textsf{PATH}, it is difficult to identify long paths; for \textsf{OUR}, it is difficult to construct long paths, since the longest path is also short and we do not increase the length of the longest path when adding edges.
Therefore, both normalized bounds become close to 1.
When $\mathit{pf}$ gets larger gradually, there are more edges in the graph.
In this stage, the lengths of paths in the graph are more diversified: some paths are short and some paths are long.
So it is relatively easy to connect short paths into long paths while keeping the length of the longest path unchanged.
Therefore, our method becomes more effective.
This explains why the data points of \textsf{OUR} in Fig. \ref{fig:pf} become decrease when $\mathit{pf} \le 0.2$.
When $\mathit{pf}$ is close to 0.5, the generated graph is more sequential.
In this stage, all paths in the graph are long.
It is difficult to connect paths by adding edges while keeping the length of the longest path unchanged.
Therefore, \textsf{OUR} becomes close to \textsf{PATH}.
This explains why the data points of \textsf{OUR} in Fig. \ref{fig:pf} become increase when $\mathit{pf} \ge 0.2$.
Compared to \textsf{PATH}, the maximum improvement for the normalized bound is 25\%.
Fig. \ref{fig:v} shows the results of changing the vertex number of DAG tasks.
Same as \textsf{PATH}, \textsf{OUR} is insensitive to the vertex number, which implies that our method can be applied to realistic applications with a large number of subtasks.
In this experiment, compared to \textsf{PATH}, the average improvement for the normalized bound is 21.6\%.

\subsection{Evaluation of Scheduling Task Sets}
\label{sec:evaluation_set}

This section evaluates the performance of the proposed approach for scheduling task sets.
The following approaches are compared.
\begin{itemize}
  \item \textsf{FED}. The original federated scheduling proposed in \cite{li2014analysis}.
  \item \textsf{PATH}. The federated scheduling approach in \cite{he2022bounding} by using the information of multiple long paths to reduce the number of cores required by an individual task.
  \item \textsf{OUR}. Our approach presented in Section \ref{sec:extension}.
\end{itemize}
As shown in \cite{he2022bounding}, \textsf{PATH} has the best performance among all existing scheduling approaches of different paradigms (federated, global, partitioned and decomposition-based, see Section \ref{sec:related}), so we only include \textsf{FED} and \textsf{PATH} in our comparison.

\textbf{Task Set Generation.}
DAG tasks are generated by the same method as Section \ref{sec:evaluation_one} with $c(v)$, $|V|$ and $\mathit{pf}$ randomly chosen in [50, 100], [50, 250], [0, 0.5], respectively.
The period $T$ (which equals $D$ in the experiment) is computed by $len(G)+\alpha(vol(G)-len(G))$, where $\alpha$ is a parameter. Same as the setting of \cite{he2022bounding}, we consider $\alpha$ in [0, 0.5] to let heavy tasks require at least two cores.
The number of cores $m$ is set to be 32 (but changing in Fig. \ref{fig:g_m}).
The \emph{utilization} of a task is defined to be $vol(G)/T$, and the utilization of a task set is the sum of all utilizations of tasks in this task set. 
The \emph{normalized utilization} of a task set is the utilization of this task set divided by the number of cores.
In the experiments, the normalized utilization $\mathit{nu}$ of task sets is randomly chosen in [0, 0.8].
To generate a task set with a specific utilization, we randomly generate DAG tasks and add them to the task set until the total utilization reaches the required value.
For each data point in Fig. \ref{fig:evaluation_set}, we randomly generate 1000 task sets to compute the average performance.

\begin{figure}[t]
\centering
\subfloat[core number]{
    \includegraphics[width=0.49\linewidth]{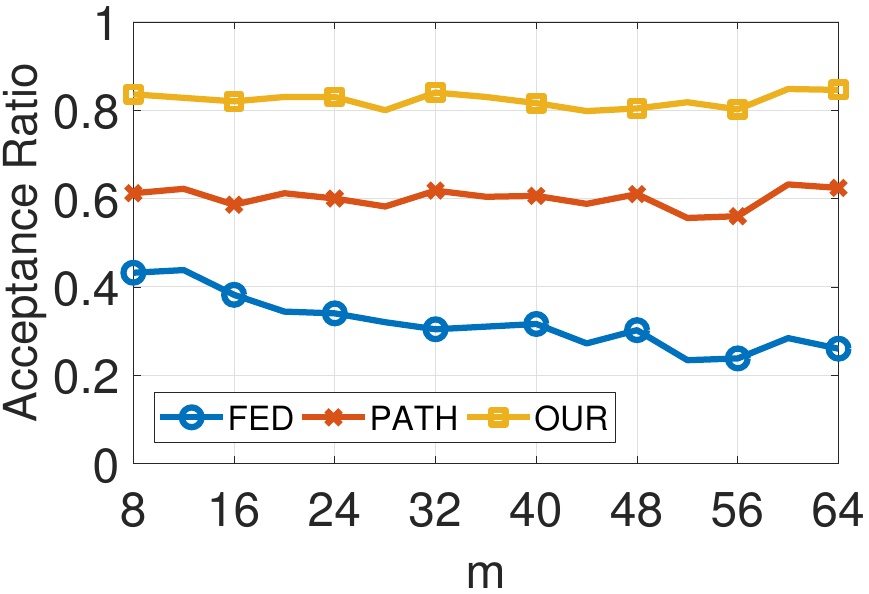}
    \label{fig:g_m}
}
\subfloat[normalized utilization]{
    \includegraphics[width=0.49\linewidth]{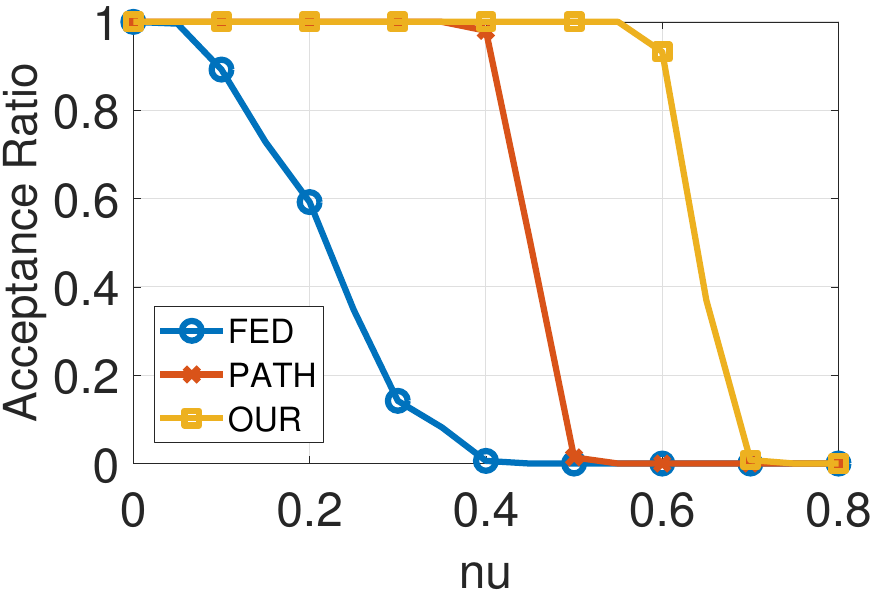}
    \label{fig:g_nu}
}
\hfil
\subfloat[deadline]{
    \includegraphics[width=0.49\linewidth]{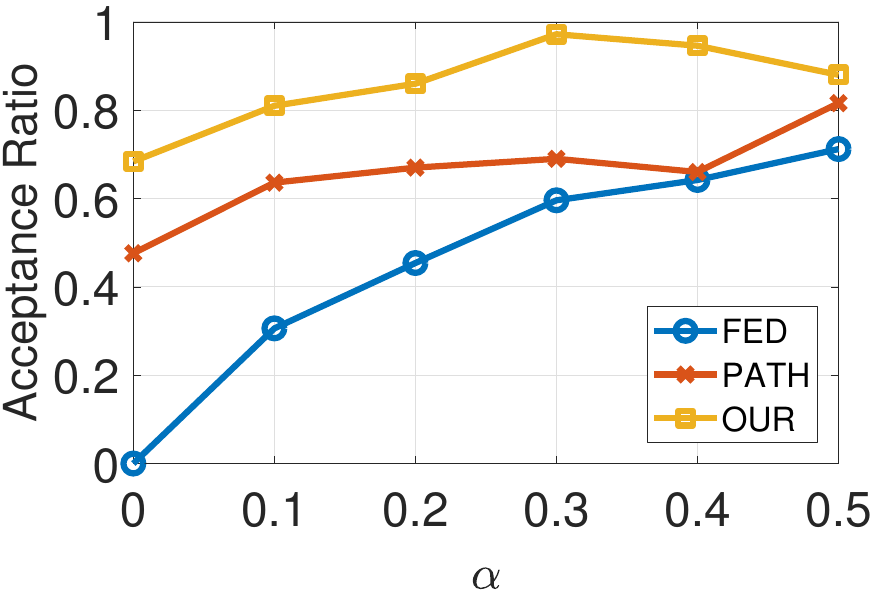}
    \label{fig:g_t}
}
\subfloat[parallelism factor]{
    \includegraphics[width=0.49\linewidth]{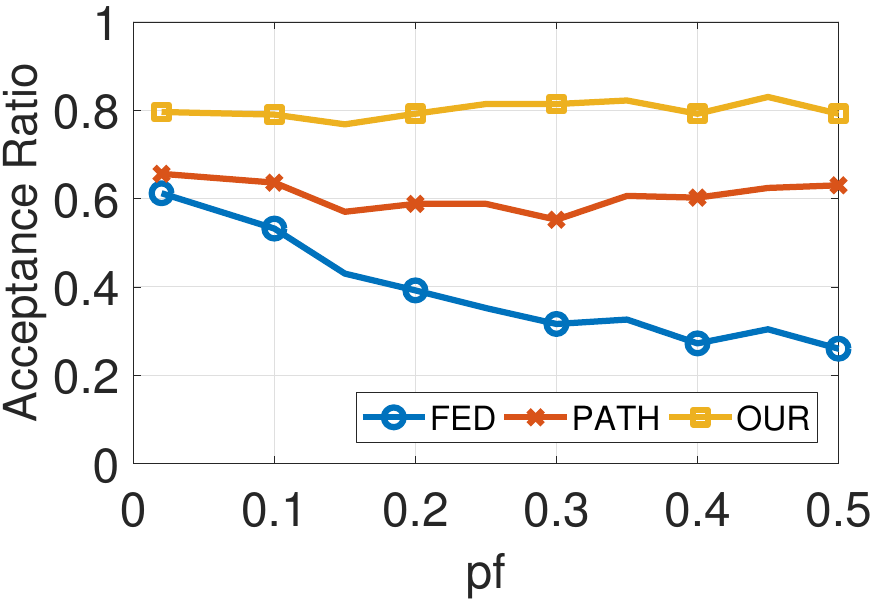}
    \label{fig:g_p}
}
\caption{Evaluation results of task sets.}
\label{fig:evaluation_set}
\end{figure}

We evaluate the schedulability of task sets using the acceptance ratio as the metric. The \emph{acceptance ratio} is the ratio between the number of schedulable task sets and the number of all task sets. The larger the acceptance ratio, the better the performance is.
The result of changing the number of cores on which task sets are scheduled is in Fig.~\ref{fig:g_m}, which shows that our approach and federated scheduling in general are insensitive to the number of cores. This is because in federated scheduling paradigm, heavy tasks are mostly allocated and executed on dedicated cores, thus different tasks cannot interfere with each other. This is important for scalability.
Compared to \textsf{PATH}, the average performance improvement of our approach is 22.2\% among all numbers of cores.
In the following, we choose $m=32$ as a representative for other experiments.

Fig. \ref{fig:g_nu} reports the results with different normalized utilizations.
Compared to \textsf{PATH}, the maximum improvement of acceptance ratio is 100\%, which means that there are task set configurations that all existing approaches cannot schedule but our approach can. This further demonstrates the effectiveness of the proposed approach.
Fig. \ref{fig:g_t} presents the result of changing $\alpha$. Different $\alpha$ means different deadlines. Fig. \ref{fig:g_t} shows similar trends as reported in \cite{he2022bounding}.
When $\alpha$ is close to 0, the deadline $D$ approaches $L$, which means that the tasks are more difficult to schedule. However, with the proposed technique of adding edges, the response time bound is greatly reduced and less likely to exceed the deadline. Therefore, tasks using our approach are more likely to be scheduled.
When $\alpha$ increases, the deadline $D$ becomes larger and close to the $vol(G)$. In this case, the performances of all scheduling approaches increase.
The maximum improvement in this experiment is 28.6\% with $\alpha=0.4$ compared to \textsf{PATH}.

Fig. \ref{fig:g_p} reports the result of changing the parallelism factor $\mathit{pf}$.
When $\mathit{pf}$ increases, there are more edges in the generated tasks, which means that the length of the longest path increases. \textsf{FED} assumes that other vertices in the task cannot execute in parallel with the longest path. When the length of the longest path increases, more computing resources are wasted and the performance of \textsf{FED} decreases. 
However, \textsf{PATH} can utilize the information of long paths and our approach can construct long paths such that different long paths can execute in parallel with each other, thus addressing this type of pessimism.
Therefore, the performances of \textsf{PATH} and \textsf{OUR} are not affected by $\mathit{pf}$.  
Due to space limitations, the results of changing the vertex number are not reported.
Same with the trends in Fig. \ref{fig:v}, our approach and federated scheduling in general, are insensitive to the number of vertices in DAG tasks.
Experiments show that the proposed approach consistently outperforms the state-of-the-art by a large margin, which is consistent with the theoretical result that our approach dominates \textsf{PATH}.

\section{Related Work}
\label{sec:related}
The most closely related work to this research is \cite{he2022bounding}, where He et al. proposed a response time bound for a DAG task using multiple long paths. Our work employs the bound of \cite{he2022bounding} as a guidance and proposes methods that systematically connect short paths into long paths by adding edges to reduce the worst-case response time bound for a DAG task.

Concerning modifying edges for DAG tasks in real-time scheduling, \cite{buttazzo2011hard} studied the timing anomaly of DAG tasks and observed that weakening the precedence constraints (i.e., removing edges) may lead to longer response time.
\cite{fonseca2017improved} considered removing edges of DAG tasks and transforming the DAG into a series-parallel graph \cite{he1987parallel} to derive a bound on the interference incurred by this task.
Different from the above works, we study the method of adding edges and focus on the response time bound and schedulability of DAG tasks.

For real-time scheduling of DAG tasks, existing approaches can be categorized into four major paradigms: federated scheduling \cite{li2014analysis, chen2016federated, baruah2015federatedconstrained, baruah2015federatedarbitrary, baruah2015federatedconditional, jiang2017semi, ueter2018reservation, jiang2021virtually, he2022bounding}, global scheduling \cite{li2013outstanding, chen2014capacity, nasri2019response, he2021response, dai2022response, he2023real}, partitioned scheduling \cite{fonseca2016response, casini2018partitioned}, and decomposition-based scheduling \cite{qamhieh2013global, saifullah2014parallel, jiang2016decomposition}.
Federated scheduling, where each DAG task is scheduled on a set of dedicated cores, is closely related to this work.
Federated scheduling was originally proposed in \cite{li2014analysis}. 
Later, federated scheduling was generalized to constrained deadline tasks \cite{baruah2015federatedconstrained}, arbitrary deadline tasks \cite{baruah2015federatedarbitrary}, and conditional DAG tasks \cite{baruah2015federatedconditional}.
A series of federated-based scheduling approaches \cite{jiang2017semi, ueter2018reservation, jiang2021virtually, he2022bounding} were proposed to address the resource-wasting problem in federated scheduling.

\section{Conclusion}
\label{sec:conclusion}
In this paper, through exploring the dependencies between vertices, we propose a method of adding edges and thus making longer paths in DAG tasks to optimize the worst-case response time bound.
We also apply the proposed techniques to the scheduling of multiple DAG tasks.
Experiments demonstrate that the proposed method significantly outperforms the state-of-the-art, reducing the worst-case response time bound by 21.6\% and improving the system schedulability by 22.2\% on average.

\bibliographystyle{IEEEtran}
\bibliography{reference}

\end{document}